\documentclass[11pt,fleqn,reqno]{amsart} 

\usepackage{amsrefs}  
\usepackage{amsthm}
\usepackage{epsfig}
\usepackage{graphicx}
\usepackage{enumerate}
\usepackage{color} 
\usepackage[toc,page]{appendix}
\usepackage[dvipsnames]{xcolor}
\usepackage{bm,bbm}
\usepackage{cancel}
\usepackage{enumitem}
\usepackage{physics}
\usepackage{amsmath,euscript}
\usepackage{amsthm}
\usepackage{amssymb}
\usepackage{graphicx}
\usepackage{mathrsfs}
\usepackage{epsf}
\usepackage{xcolor}
\usepackage{psfrag}
\usepackage{enumerate}
\usepackage{amsfonts,amssymb,amsthm,amsmath} 
\usepackage{hyperref}

\usepackage{enumitem} 
\usepackage{xspace}
\usepackage[margin=1.07in]{geometry}

\usepackage{todonotes} 
\usepackage{float} 
\usetikzlibrary{decorations.markings}

\usepackage{soul}

\newtheorem{theorem}{Theorem}[section]
\newtheorem{corollary}[theorem]{Corollary}

\newtheorem{lemma}[theorem]{Lemma}
\newtheorem{prop}[theorem]{Proposition}

\theoremstyle{remark}

\theoremstyle{definition}

\numberwithin{thm}{section}
\numberwithin{equation}{section}

\usepackage{soul}
\setstcolor{red}

\definecolor{green}{rgb}{0.0, 0.5, 0.5}
\definecolor{yellow}{rgb}{0.5, 0.5, 0}
\definecolor{lgray}{gray}{0.9}
\definecolor{llgray}{gray}{0.95}
\definecolor{lllgray}{gray}{0.975}

\newcommand{\nc}{\newcommand}

\nc{\la}{\label}
\nc{\ba}{\begin{array}}
\nc{\ea}{\end{array}}
\nc{\bs}{\begin{split}}
\nc{\es}{\end{split}}

\newcommand{\R}{\mathbb{R}}
\newcommand{\C}{\mathbb{C}}

\newcommand{\cB}{\mathcal{B}}

\newcommand{\cH}{\mathcal{H}}

\newcommand{\cP}{\mathcal{P}}         

\newcommand{\A}{{\rm{A\,}}}
\newcommand{\B}{{\rm{B\,}}}

\nc{\al}{\alpha}
\nc{\Sig}{\Sigma} 
\nc{\G}{\Gamma}
\nc{\et}{\eta} 
  
\nc{\g}{\gamma}
\nc{\gam}{\gamma}
\nc{\ka}{\kappa}
\nc{\lam}{\lambda}
\nc{\Lam}{\Lambda}
\nc{\Om}{\Omega}
\nc{\om}{\omega}
\nc{\ta}{\tau}
\nc{\w}{\omega}
\nc{\io}{\iota}
\nc{\z}{\zeta}
\nc{\s}{\sigma}
\nc{\sig}{\sigma}
\nc{\Si}{\Sigma}
\nc{\vphi}{\varphi}
\newcommand{\eps}{\epsilon}

\nc{\bP}{\bar{P}}
\nc{\bQ}{\bar{Q}}

\nc{\ran}{\rangle}
\nc{\lan}{\langle}

\newcommand{\ra}{\rightarrow}

\newcommand{\gs}{\gtrsim}
\newcommand{\one}{\mathbf{1}}

\newcommand{\Ran}{\operatorname{Ran}}

\newcommand{\supp}{\operatorname{supp}}

\renewcommand{\Im}{\mathrm{Im}} 

\newcommand{\im}{{\rm Im}}

\newcommand{\dist}{\mathrm{dist}}

\nc{\bfone}{{\bf 1}}
\nc{\1}{{\bf 1}}

\newcommand{\n}{\nabla}

\newcommand{\DETAILS}[1]{}

\newcommand{\AB}{{\rm{A B\,}}}

\makeatletter
\let\process@citelist\process@citelist@unsorted
\makeatother

\begin{document}
\title[Light-cone structure of  propagation of entanglement]{Light-cone structure of  propagation of entanglement} 

\author[I.~M.~Sigal]{I. M. 
Sigal}
	\address{Department of Mathematics, University of Toronto, Toronto, ON M5S 2E4, Canada }
	\email{im.sigal@utoronto.ca}
	
	\date{\today}


\maketitle
{
}

\begin{abstract}

For a wide class of bipartite systems with localized couplings, we establish existence of an effective 
 light-cone for propagation of entanglement.   
This result yields a hard lower bound on the time it takes,  under ideal conditions (no loss, no decoherence), to transport 
entanglement to  a distant location (say, a node of a graph-structured quantum network), or to maintain it there.  
  \end{abstract}



\section{Introduction}\label{sec:intro}
\subsection{Background and informal discussion}\label{sec:backgr}
Entanglement is a striking 
quantum phenomenon and a key resource in quantum information science  enabling  quantum communication (see  \cite{Acin, Alba, Amico, EisCrPlen, Flann, Gisin, Goold, Hoke, Horod, Kimble, LiuNoisy, Pir, PlenV, Pres1, Pres2,  RangTak, Sidhu, Terh, Thak, Wehner,  Xav, ZhangScie}, for reviews, and \cite{Haya, NC, OV, Pres, Watr, Wild}, for  books and lecture notes). It is also responsible for the main obstacle in quantum information processing and quantum computing - the decoherence (see \cite{StrelAdPl} for a review).

 \vspace*{.1cm}

 In experiments and applications,  entanglement  is transported by either  photons (common and at large distances) or   particles (ions, electrons, holes, etc.), 
 at distances up to several meters. The latter information transfer mode, of interest here,  has clear advantages at shorter distances by removing an intermediary, which requires additional controls and is vulnerable to loss, and  having incoming carriers  
 integrated directly into active algorithms, 
see \cite{Barrett, Bloch, Bowler, Deuar, Duan, Gross, Hasegawa, Hofstetter, Herrmann, Jaskula, Kheruntsyan, Olmschenk, Peise, Perrin, Rauch, Riebe, Rowe, Wan} for some landmark experiments.

\vspace*{.2cm}

Theoretically, viability of quantum communication is determined by the propagation time of information (modulo an exponentially small leakage)
versus the decoherence time. For this, one needs an estimate of the effective speed of propagation of the former.
 However,  until now, even the finiteness of the propagation speed for entanglement has not been proven. 

\vspace*{.1cm}

 The reason for this disconnect 
 goes to the root of this subject:  the Lieb-Robinson bounds - the common tool in proving the finiteness of speed of propagation in quantum information theory - estimate commutators of propagating observables, while there is no known characterization of entanglement in terms of any system of observables.

\vspace*{.1cm}

In this paper, we present the first proof of the finiteness of speed of propagation of  the particle mediated entanglement, establishing the effective light cone structure for this important  quantum information   resource. 
We also produce an estimate of the maximal propagation speed.
 
\vspace*{.1cm}

In the case of photon transported 
 entanglement, a natural framework for a rigorous treatment is provided by the non-relativistic QED (see e.g. \cite{BoFauSig, GS}), but such a treatment is still missing. 
\vspace*{.1cm}

 We establish the finiteness of speed of propagation of entanglement for bipartite particle systems with localized couplings between their 
 subsystems yielding   the entanglement light-cone structure. Our 
  approach 
 is new and simple. It is  indicated at the end of Subsection \ref{sec:result}.
\vspace*{.1cm}

  The physical set-up we consider is of systems $A$ and $B$ with the state spaces $\mathcal{H}_{A}=L^2(\Lam)$, and $\mathcal{H}_{B}$, respectively, coupled in a subset $Y$ of $\Lam$, where $\Lambda$ is a domain in either $\mathbb R^n$ or $\mathbb Z^n$.  We do not specify the nature of system $B$.   It could be of the same physical nature as $A$ (say, a few-particle system), or an 
 ancilla 
 system  used for intermediate operations  
 and to be discarded later on, or an environment absorbing information (and energy) flow.

\vspace*{.1cm}

The entanglement is tested in a domain $X\subset \Lam$ at time $t$  (see 
 Fig. 1 below).

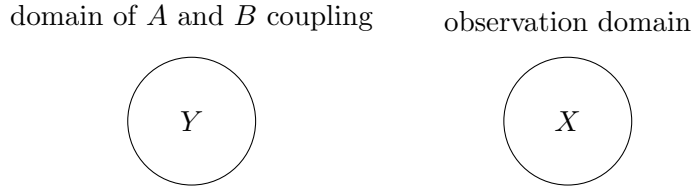
\begin{figure}[H]
    \centering
    \begin{tikzpicture}[
  >=latex,
  every node/.style={font=\small},
  set/.style={draw,circle,minimum size=1.7cm},
  squig/.style={->,decorate,decoration={snake,amplitude=1.5pt,segment length=6pt}}
]

  \node[set] (Y) at (0,0) {};
  \node[set] (X) at (5,0) {};

  \node at (Y) {$Y$};
  \node at (X) {$X$};

  \node[above=2mm of Y] {domain of $A$ and $B$ coupling};
  \node[above=2mm of X] {observation domain};

\end{tikzpicture}
    \caption{\small Physical set-up: entanglement is created at $t=0$ in $Y$ and measured at $t$ in $X$.}
    \label{fig:set-up}
\end{figure}
\noindent 
 Such a set-up is used  in many experiments studying and using entanglement, see e.g. the set of references in the second paragraph on page 1. 

\vspace*{.1cm}

 As usual,  the state space   of the compound (bipartite)  system $AB$ is given by
  \[\mathcal{H}_{AB}=\mathcal{H}_{A}\otimes\mathcal{H}_{B}\]
and the evolution of  $AB$ is provided by the von Neumann equation  for density operators (i.e. positive, trace-class operators of trace $1$) on $\mathcal{H}_{AB}$, \begin{equation}\label{vNE-eqAB}
 \frac{\partial  \G_t}{\partial t}= -i  [H_{\AB},  \G_t],\ \quad  \G_{t=0} =\G_0, 
\end{equation}
with $H_{AB}$ the standard bipartite  Hamiltonian on $\mathcal{H}_{AB}$,  described below. 

Here and in what follows, we use the units with $\hbar=1$,  the speed of light $=1$ and the electron mass $=1$. 

We formulate our main result. It uses    
 the notions of local entanglement, local separability (mod $\kappa\in (0, 1)$),  introduced in the next subsection. 
  The latter notion describes states at the trace norm distance  $\kappa$ from separable ones.   
  Let $d_{XY}$ denote the distance between subsets $X$ and $Y$ of $\Lam$.   
   We show

\begin{theorem}\label{cor:ent-prop1}  Assume Conditions  \eqref{HAHB-cond} - \eqref{I-bndd'} and (AH) below are satisfied and let  $a>0$ be as in Condition (AH), $\mu\in (0, a)$ and $c(\mu)>0$ be defined in \eqref{cmu} below.  
We have

(a) If  at time $t=0$,  
the system $AB$ is separable outside a set $Q\subset \Lam$,  then  it is separable mod $\kappa$ in any set $X\subset Q^c$  for time \[t<d/c,\ d:=\min (d_{XY}, d_{XQ}),\  c> c(\mu),\]  
 where  $\kappa:=O(e^{-2\mu(d-ct)})$, provided  $d$ is sufficiently large. 
 
 (b)  If  at time $t=0$,  
the system $AB$ is entangled in a set $Q\subset \Lam$, with $\G_0$ of  finite rank, 
 then it stays entangled  in $X\supset Q$  for time 
 \[t<d'/c,\ d':=\min (d_{XY}, d_{X^cQ}),\ c>  c(\mu),\] provided  $d'$ is sufficiently large. \end{theorem}

This theorem is proven in Section \ref{sec:bipart-evol-pf}. Part (a) says that one cannot transport entanglement from $Y$ to $X$ if $X$ lies outside the  forward {\it light cone} 
\begin{align}\label{LC}\Lam_{Y, c}:=\{(x, t): d_Y(x)\le ct\}, \qquad d_Y(x)=d(x,Y),\end{align}
of $Y$ for the speed $c=c(\mu)$ (see Fig. \ref{fig:lcb}). Put differently, it takes at least time $d/c(\mu)$  to transport entanglement from $Y$ to $X$. According to part (b), the system $AB$ retains entanglement in a  domain $Q$ for time at least  $d'/c(\mu)$. 
 To put differently, the entanglement concentrated initially in a set $Q$ stays within the light cone, $\Lam_{Q, c}$ of $Q$.


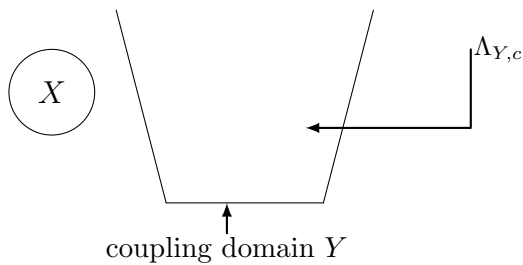
\begin{figure}[H]
\centering
\begin{tikzpicture}[line cap=round, line join=round, scale=0.95]

\tikzset{
  callout/.style={->,>=latex,thick},
  labelbg/.style={fill=white, inner sep=1pt, rounded corners=0.5pt}
}

\coordinate (A) at (0,0);
\coordinate (B) at (2.2,0);
\coordinate (L) at (-0.7,2.7);
\coordinate (R) at ( 2.9,2.7);
\draw (A)--(L);
\draw (B)--(R);
\draw (A)--(B);

\draw (-1.6,1.55) circle [radius=0.6];
\node[labelbg] at (-1.6,1.55) {$X$};

\coordinate (supppt) at (0.85,0);
\coordinate (ypt)    at (1.55,0);

\node[labelbg] (supptxt) at (0.85,-0.65) {\small coupling domain $Y$};
\draw[callout] (supptxt.north) -- (supppt);


\coordinate (coneInt) at (1.95,1.05);

\node[align=left, font=\footnotesize, labelbg] (lamtxt) at (4.65,2.15)
  {$\Lambda_{Y, c}$}; 

\draw[callout] (lamtxt.west) |- (coneInt);

\end{tikzpicture}

\caption{Light cone $\Lambda_{Y, c}$  of $Y$. If a set $X$ lies outside $\Lambda_{Y, c}$, then the probability of entanglement in $X$ is negligible.}
\label{fig:lcb}
\end{figure}

 
To our knowledge, this is the first result on the space-time dynamics of the entanglement.

\vspace*{.1cm}

This  rigorous result, obtained from the first principles, yields a hard lower bound, $T_{\rm min}\ge$ $L/c(\mu)$, on time it takes,  under ideal conditions (no loss, no decoherence), to transport entanglement to a location at distance $L$. 
It informs one about 
the transport   
of entanglement over  
quantum networks 
and sets hard constraints on operating times  
and scalability  in complex quantum devises. 

For comparison, we list recent results on the entanglement dynamics,  \cite{BHV, Cala1, Cala2, Jonay, Keys, KimHuse, Li, Mez, Nah1, Nah2, Nah3, Skinn, Swann, Ton, TranEtal2} where one can also find references to earlier works.

\vspace*{.1cm}

Our results are complementary to  the Lieb-Robinson bounds (LRB) on propagation of correlations of observables  
(\cite{BHV, CEPH,  CL, CLY,  EldredgeEtAl, EisOsb, EpWh, FLS2, Foss, H1, H3, KGE, KuwLem, KuwS1, KuwS2, KVS, LRSZ,   LR, MatKoNaka, NachOgS, NachSchlSSZ, NachS, NSY2,  SHOE, SZ,  TranEtal1, TranEtal2, WaH}).
 Our approach is also complementary  
 to approaches developed the cited works.
(For spin systems and fermi gasses, the LRB is uniform in the states and is of the same form for entangled and separable states. For bose gases, it depends on the space distribution of the states but not on their correlations (see \cite{FLS2, KuwLem, KuwS1, KuwS2, KVS, LRSZ, LRZ}).

In the rest of the introduction, we give a rigorous definition of the local entanglement, describe precisely the model, and formulate our main technical results, Theorem \ref{thm:Gt-est} and \ref{thm:SN-est}, which in particular  imply Theorem \ref{cor:ent-prop1}. After that, we  discuss introduced notion and our results. Theorems \ref{cor:ent-prop1}, \ref{thm:Gt-est} and  \ref{thm:SN-est} 
   are proven in Section \ref{sec:bipart-evol-pf}.  
   The latter section is rather technical but  fairly short and straightforward. In Appendix \ref{sec:tech-lem}, we prove some elementary technical statements. In Appendix \ref{sec:SN-prop}, we prove some elementary properties of the Schmidt number. The latter plays an important role in this paper.
   
In what follows, we measure distance between density operators in terms of  the trace norm  $\|\lambda\|_{S_1^{\AB}}=\Tr_\AB (\lambda^*\lambda)^{\frac{1}{2}}$.  Physical significance of the trace norm is discussed in Remark 2 below.

\vspace*{.2cm}

\subsection{Local entanglement}\label{sec:loc-ent} 

Let $\chi_X$ denote for the characteristic function of $X\subset \Lambda$ and $\one_\#$ be the identity operator operator on the space $\cH_\#,\ \#=A, B$. 
 
 \vspace*{.1cm}
 
 We say that  $\Psi\in\mathcal{H}_{AB}$ is {\it separable} 
 in $X$   iff $(\chi_X\otimes \one_\B)\Psi$ is separable, i.e.

\[(\chi_X\otimes \one_\B)\Psi= \psi \otimes \phi\ \text{ for some }\ \psi \in\mathcal{H}_{A}\ \text{  and }\ \phi \in\mathcal{H}_{B},\]
and is {\it entangled}, otherwise. For $X=\Lam$, this gives the standard definition of entanglement for pure states.

Similarly, we say that a general (pure or mixed) state $\G$ on $\cH_{\AB}$ is  {\it separable}   in $X$ iff $\tilde\chi_X(\G )$, where     
\begin{equation} \label{tilde-chi} 
\tilde\chi_X(\G)={(\chi_X\otimes \one_\B)}\G{(\chi_X\otimes \one_\B)},\end{equation} 
is separable, i.e it is    an incoherent 
 mixture of product states, 
\begin{align}\label{un-ent-do}\tilde\chi_X(\G )= \sum_i p_i\g_{\A}^i\otimes \g_{\B}^i,\end{align}
 for some  weights $p_i>0, \sum_i p_i=1,$ and some  density operators $\g_{A}^i$ and $\g_{B}^i$ of systems $A$ and $B$, or a $S_1^\AB$-limit of such states in the trace norm topology.   
 Correspondingly, a state $\G $ is said to be {\it entangled} in $X$ iff $\tilde\chi_X(\G )$   is entangled, i.e. is not separable.

For correlations and, in particular, for pure states, one replaces \eqref{un-ent-do} by $\tilde\chi_X(\G )=  \g_{A} \otimes \g_{B}$.  
    
\vspace*{.1cm}

 Denote by  $ \Sigma_{\#}$  the set of density operators on $\mathcal{H}_\#$, $\#=\A, \B$, and by $\Sigma_{\rm sep}$ the set of separable  states on $AB$. The latter is defined as   (\cite{Wern, HSW}) 
 \begin{align}\label{sep-set'} \Sigma_{\rm sep} &:= \overline{\text{conv} \{ \g_\A \otimes \g_\B : \g_\A \in \Sigma_\A, \g_\B \in \Sigma_\B \}}, 
 \end{align}
where $\overline{\text{conv} X}$ denotes the closure  in the trace norm topology of the set of all convex combinations of elements of $X$. For a discussion of this space, see Subsection \ref{sec:SN} below.
  
 We say a state $\G$ is  separable mod $\kappa$ 
   in $X$, $\kappa\in (0, 1)$,  iff 
 \begin{align}\label{ka-ent}	\dist_{S_1}^X(\G, \Sigma_{\rm sep})\leq \kappa, \end{align}
where $\dist_{S_1}^X(\G, \Sigma_{\rm sep}):=\inf_{\G'\in \Sigma_{\rm sep}} \frac12 \| \tilde\chi_X(\G-\G')\|_{S_1^{\AB}}$. 
This definition is refined in Remark 4 below.  
   
\vspace*{.1cm}

 Finally, we say that $\G\in S_1^\AB$  
 is  supported in a set $Q\subset \Lam$ iff  $\G =\tilde\chi_Q(\G)$. 
 
 \vspace*{.1cm}

Below, we use the following notation. $\Sigma$ and $\mathcal B$ will denote the spaces of  density operators ($AB$-states) on $\mathcal{H}_{AB}$ and   bounded operators ($A$-observables) on    $\mathcal H_A$;  $\mathcal S_1^\#$ will stand for the Schatten space of trace-class operators on $\mathcal H_\#$ for  $\#=A, B, AB$, see e.g. \cite{Schatten, Sim}. The norms  in $\mathcal S_1^\#$ are denoted  by $\|\cdot\|_{\mathcal S_1^\#}$. $P_\phi=|\phi\ran\lan \phi|$ denotes the orthogonal projection onto the one dimensional space spanned by $\phi$. For positive operators $A$ and $B$, the notation $A \gs B$ signifies that there is a constant $C$, independent of essential parameters, s.t. $A \ge C B$. Recall that we use  the units with $\hbar=1$. Finally, we use the terms `density operators' and `states' interchangeably.

\vspace*{.2cm}

\subsection{Model and results}\label{sec:result}

We  consider  the standard bipartite  Hamiltonian
\begin{align}\label{H-AB}	H_{AB}=	H_{\A}\otimes \one_\B+\one_\A\otimes H_{\B} + I=H_0+I, 
\end{align}
where $H_{0}:=H_{\A}\otimes \one_\B+\one_\A\otimes H_{\B}$. We assume that the Hamiltonians $H_{\A}$ and $ H_{\B}$ of systems $A$ and $B$ and the interaction operator $I$ satisfy the conditions
\begin{align}\label{HAHB-cond}	&H_{\A}, H_{\B}\ \text{  are self-adjoint and non-negative, } H_{\A}, H_{\B}\ge 0;\\
\label{I-loc}	&I=\tilde\chi_Y (I)\ 
	 \text{  for some set }\ Y\subset \Lambda;\\
\label{I-bndd}	&\|(H_{0}+1)^{-1/2}I (H_{0}+1)^{-1/2}\|< 1,\\
\label{I-bndd'}&\|I (H_{0}+1)^{-1}\|< 1. 
\end{align}
  Condition \eqref{I-bndd} is stronger than necessary and is assumed for simplicity. It could be relaxed.


Moreover, let $S_a^n:=\{\zeta= ( \z_1, \dots, \z_n)\in \C^n: |\Im \z_j|<a\ \forall j\}, a>0$. With this,  we assume  that the operator $H\equiv H_{\A}$ satisfies the following condition:

\bigskip

\begin{itemize}
\item[(AH)] The operator-function $H_\xi:=e^{i\xi\cdot x}H e^{-i\xi\cdot x}, \xi \in \mathbb R^n, x\in \Lam$, has an analytic continuation, $H_\zeta$, from $\mathbb R^n$ to $S_a^n$ s.t., $\forall\, \eta\in \R^n$ s.t.   $i\eta\in S_a^n$,  
\begin{equation}\label{cond: AH}\text{ 
$\im H_{i\eta}:=\frac{1}{2i}(H_{i\eta}-H_{i\eta}^*)$ extends to a bounded operator}. 
\end{equation}
\end{itemize}

We discuss  Condition (AH) in Remarks 8 and 9 and Subsection \ref{sec:H} below.  For any bounded operator $B$, we  let $\sup B=\sup_{\psi\in D(B), \|\psi\|=1}\lan \psi, B\psi\ran.$   Due to \eqref{cond: AH}, the following function is well defined:
 \begin{equation}\label{cmu}    c(\mu):= \sup\limits_{\z\in \overline{S_\mu^n}}  \sup\im\ H_{\z}/\mu, \qquad \mu < a. 
\end{equation}
It is shown in \cite{SigWu2} (Proposition B.1) that $c(\mu)>0$, $ \forall \mu\in (0,a)$, provided the Heisenberg (group) velocity does not vanish, $i[H, x]\neq 0$. For a computation of $c(\mu)$, see Remark 9 and Subsection \ref{sec:H}).

\vspace*{.1cm}

\vspace*{.1cm}

Recall the notation $H_{0}:=H_{\A}\otimes \one_\B+\one_\A\otimes H_{\B}$ and let $L_0  \G := -i  [H_{0},  \G]$. We have  

\begin{theorem}\label{thm:Gt-est} Assume Conditions \eqref{HAHB-cond} - \eqref{I-bndd'} and (AH) and let $\G_{t}$ satisfies \eqref{vNE-eqAB}  with an initial condition $\G_{0}$ satisfying   
\begin{align}\label{Gam0-cond}\|H_{0}\G_0 \|_{S_1^{AB}}<\infty. 
\end{align}
 Then, for any $\mu\in (0,a)$ and for any set  $X$ in $\Lambda$ disjoint from $Y$, the evolution $\G_t$ satisfies
\begin{align}\label{Gt-est}
    \|\tilde\chi_X \big( \G_t - e^{t L_{ 0}}\G_0 \big)\|_{S_1^{AB}}\leq Ce^{- \mu(d_{XY}-ct)}\|(H_{0}+1)\G_0 \|_{S_1^{AB}},\end{align}
for any $c>c(\mu)$ and some constant $C=C_{n,c,\mu}>0$ depending on $n,c,\mu$ but {\it independent} of $\G_0$. Here  $c(\mu)$ is given  by~\eqref{cmu}.
\end{theorem}

This theorem is our first main technical result. It is proven in Section \ref{sec:bipart-evol-pf}. We use it there  to prove Theorem \ref{cor:ent-prop1}.  The proof of  Theorem \ref{cor:ent-prop1}(a) is rather direct, while the one of  Theorem \ref{cor:ent-prop1}(b) is more delicate. Here, we use 
  the notion of Schmidt number (\cite{HorTerh, Shir}),  $n_S(\G)$, defined in Subsection \ref{sec:SN}. For now, we mention only that it extends the notion of  Schmidt  rank (SR),  $r_S(P_\Psi)$, from pure to mixed states
while retaining the key property of the SR: 
\begin{align}\label{SR-criter} \text{pure state }  P_\Psi\ \text{ is entangled} \iff r_S (\Psi)>1.\end{align}

\vspace*{.1cm}

 We say that  a state $\G \in\mathcal{S}_1^{\AB}$  has  the SN $k$  in a domain $X$ if and only if the (not normalized) state $\tilde\chi_X(\G )$  
 has  the SN $k$. We have

\vspace*{.1cm}

\begin{theorem}\label{thm:SN-est} Let the conditions of Theorem \ref{thm:Gt-est}
 be satisfied and  let $\G_{t}$ and $\G_{0}$ be as in Theorem \ref{thm:Gt-est}. Assume  $\G_0$ has the SN $k$ in some domain $Q\subset \Lam$ and is of a finite rank.    Then, for any $X\supset Q$,  \begin{align}\label{SN-est}n_S (\tilde\chi_X(\G_t)) \ge k, \end{align}  
for the time $t<d'/c(\mu)$, with $d'=\min (d_{XY}, d_{X^cQ})$, provided $d'\gg1$. 
 \end{theorem}


This theorem is our second main technical result. It is proven in Section \ref{sec:bipart-evol-pf}. Theorem \ref{cor:ent-prop1}(b) follows from Theorem \ref{thm:SN-est} and the following result proven  (following results of  \cite{HSW, Shir}) in the next subsection. 
\begin{theorem}\label{thm:SN-criter}  
\begin{align}\label{SN-criter-fr} \G\ \text{ is entangled} \iff n_S (\G)>1;\ \G\ \text{ is separable } \iff  n_S (\G )=1.\end{align}\end{theorem}

   \vspace*{.1cm}

  Theorems \ref{thm:SN-est} and \ref{thm:SN-criter} imply  Theorem \ref{cor:ent-prop1}(b).
  
The proof  of Theorem \ref{thm:Gt-est} is based on the Duhamel principle and a bound on Schr\"odinger propagators of \cite{SigWu}. The proof  of Theorem \ref{thm:SN-est} uses, besides Theorem \ref{thm:Gt-est},
properties of the Schmidt number proven in Appendix \ref{sec:SN-prop}. 

\noindent{\it Discussion.} By \cite{BHV}, exponentially small leakages 
  outside the light cone carry at most exponentially small amount of information.


\subsection{Schmidt number} \label{sec:SN}

In this subsection, we define the Schmidt number (SN) $n_S$ and prove Theorem \ref{thm:SN-criter}. To ease the reading, we first define the SN in  
the finite-rank case.

  Recall, $\cP:=\{P_{\Psi}: \Psi\in \cH_\AB\}$,  the metric space  of pure states, $P_\Psi=|\Psi\ran\lan \Psi|$.   Let $M$ denote the set of positive functions, $\mu:\cP\ra [0, 1]$, with countable  supports 
  in $\cP$, satisfying $\sum_\psi \mu(\psi)=1$, i.e. $\mu$ are discrete probability distributions. 
For such a function $\mu\in M$, we define the map $\cB_{\rm fr}: M \ra \Sigma $ (here `fr' stands for the finite rank) by  
\begin{align}\label{Gam-mu-fr}\cB_{\rm fr}: \mu\ra \G(\mu) := \sum_{\Psi\in \mathcal P} \mu(\Psi) P_{\Psi}. \end{align}

\vspace*{.1cm}

  Now, for a state $\G$ of finite rank, 
   we define the Schmidt number  as   (see  \cite{HorTerh, Shir})
\begin{align}\label{SN-fr}
n_S^{\rm fr} (\G) := \inf_{\mu\in M_\G}  \operatorname*{sup}_{\Psi\in \supp\mu} r_{S}(\Psi), \end{align}
where $M_\G:=\{\mu\in M : \cB_{\rm fr}(\mu)=\G\}=\cB^{-1}_{\rm fr}(\G).$ If $\G=P_{\Psi}$, then $\supp \mu=\{\Psi\}$ and $n_S^{\rm fr} (\G)=r_S(\Psi)$.

 \vspace*{.1cm}

 In general case, we let $M(\Om)$ denote the set of Borel probability measures on a metric space $\Om$.  For a measure  $\mu\in M(\mathcal P)$,  one defines (instead of \eqref{Gam-mu-fr}) the map $\cB: M(\mathcal P)\ra \Sigma$ by  (\cite{HSW}) 
\begin{align}\label{Gam-mu}\cB: \mu\ra \G(\mu) := \int_{\mathcal P} P_{\Psi}\, d\mu(\Psi), 
\end{align}
where $P:\Psi\ra P_{\Psi}$ and the integral is understood in the {Bochner} sense. (Here, we write $d\mu(\Psi)$ for $d\mu(P_{\Psi})$.) 
The map $\cB$ 
is many-to-one. Since  the space $\Sigma$ is compact and convex with the set of extreme points given by the space $\cP$,  by a version of Choquet theorem, due to \cite{HSW}, $\cB$ is onto: $\Sigma =\cB(M(\mathcal P))$. 

If $\G=\sum_i p_i P_{\Psi_i}$, then $\mu =\sum_i p_i \delta_{\Psi_i}$ (the atomic measure supported on $\{\Psi_i\}$).
 
   \vspace*{.2cm}

  Now, for a (mixed) state $\G$, one defines the Schmidt number  as   (see  \cite{Shir}) 
  \begin{align}\label{nS}
n_S(\G) := \inf_{\mu\in M_\G}  \operatorname*{sup}_{\Psi\in \supp\mu} r_{S}(\Psi), 
\end{align}
where $M_\G:=\{\mu\in M(\cP): \G(\mu)=\G\}=\cB^{-1}(\G).$ Here, we use the notation 
\[
\operatorname*{sup}_{\psi\in \supp\mu} f(\psi):= \inf\big\{ M \ge 0 :
f(\psi)\le M \text{ for $\mu$-almost all } \psi\big\}.\]
Set  $n_S(\G)=0$ for $\G =0$.  Note that $n_S(P_{\Psi})=r_S(\Psi)$.

   \vspace*{.1cm}
  In conclusion of these preliminaries, we observe that the space $\Sigma_{\rm sep}$ is compact and convex with the set of extreme points given by the space $\cP_{\text{sep}}:=\{P_{\Psi}: \Psi\in \cH_\AB {\text{ is separable}}\}$ of pure separable  states, $P_\Psi$, on $AB$. By the Krein-Milman 
  theorem, we have
\begin{align}\label{sep-set''} \Sigma_{\rm sep} = \overline{\text{conv}\ \cP_{\text{sep}}}. \end{align} 
Moreover, let  $M_\AB$  be the set of   {Borel probability measures on  the 
  space}  $\cP_{\text{sep}}$.  
It is shown in \cite{HSW} 
 that,  like $\Sigma =\cB(M(\mathcal P))$, we have 
 \begin{align}\label{sep-set} 
 \Sigma_{sep}  &= \{\int_{\cP_{\text{sep}}}  P_{\Psi} \, d\mu(\Psi),\ \forall \mu \in M_\AB\}  =\cB(M(\mathcal P_{\text{sep}})). 
 \end{align}

\begin{proof}[Proof of Theorem \ref{thm:SN-criter}]  
  
   \vspace*{.1cm}

 Define the first Schmidt class (\cite{HSW}) $\Sigma_1:=\{\G\in S_1^\AB: n_S(\G)\le 1\}$ and let $\cP_1:=\{P_{\Psi}: \Psi\in \cH_\AB, r_S(\Psi)\le1\}$.  Since the SN $n_S$ is subadditive  and lower semi-continuous (by  Proposition 1 of  \cite{Shir}, see also Proposition \ref{prop:SN-prop} of Appendix \ref{sec:SN-prop} for an elementary proof for the finite rank case),  the set $\Sigma_1$ is convex  and closed (and therefore, compact) and, as it is not difficult to check, the set of extreme states of $\Sigma_1$ is  $\cP_1$, which is also closed by the fact that the rank $r$ is also  
 lower semi-continuous.   
  Hence,  by an enhanced version of the  Choquet  theorem, due to \cite{HSW},  we have  (see  Proposition 3 of  \cite{Shir})  
  \begin{align}\label{Sig1-M1} \Sigma_1 =\cB(M(\mathcal P_1)).
 \end{align}

   \vspace*{.1cm}

  Since $\Sigma=\{\G \in S_1^\AB: \G\ge 0, \Tr\G=1\}$, we have  $\Sigma_1 =\{\G\in \Sigma: n_S(\G)=1\}$ and $\cP_1 =\{P_{\Psi}: \Psi\in \cH_\AB, r_S(\Psi)=1\}$ and therefore, by Eq. \eqref{SR-criter},  
   $\cP_{\text{sep}}=\cP_{1}$. Thus, by Eqs \eqref{sep-set} and \eqref{Sig1-M1},  
 $\Sigma_{\rm sep}=\Sigma_1$, which implies Theorem \ref{thm:SN-criter}. 
\end{proof}
 
\subsection{Remarks}\label{sec:rem} 
1.  The system $B$ could be of infinite degrees of freedom and in the thermodynamic regime, but, for simplicity, we consider only normal states of $AB$, i.e. those given by density operators on $\cH_\B$. 
\vspace*{.1cm}

2.  The trace norm bounds the expectations for all observables as follows from the relation $\|\lambda\|_{\mathcal S_1}=\sup_{ 
\|A\|=1} |\Tr(A\lambda)|$. Consequently, one can interpret the trace distance $\|\lambda-\s\|_{\mathcal S_1}$ as measuring the physical distinguishability of states.

\vspace*{.1cm}

 3.  In  \eqref{ka-ent}, the distance of a given density operator $\G$ to the set  $\Sigma_{\rm sep}$ is measured in the trace norm. 
It implies also a bound on the Bures distance $D(\rho,\sigma):=\sqrt {2(1-F(\rho ,\sigma ))}$ {(which is an extension of the Fubini-Study distance to mixed states)} defined by  the fidelity, $F(\rho,\sigma):=\|\sqrt{\rho}\sqrt{\sigma}\|_{1}$, which has the important property of monotonicity under local operations and classical communication (LOCC)  
 (see e.g. \cite{FaRe}, Appendix B, \cite{Horod} and\cite{NC}). 

\vspace*{.1cm}

   4. 
 For a  further refinement of the definition of the  $\kappa$-separable state given above, we   say that $\G$ is  entangled in $X$ at the space-time correlation length $1/\mu$ and speed $c>0$ iff \eqref{ka-ent} holds with $\kappa= O(e^{- \mu(d_{XY}-ct)})$.  

\vspace*{.1cm}

5. 
 Local entanglement and  separability are not binary properties. There are properties in between, like  semi-localization mentioned after Eq. \eqref{Gam''}.

 \vspace*{.1cm}

6. One can apply the standard entanglement measures to space-localized states, $\G_X:=\tilde\chi_X(\G)/\Tr(\tilde\chi_X(\G))$, to obtain  space-local  entanglement measures,  $E_X(\G):=E (\G_X)$.

\vspace*{.1cm}

7.  We say that the operator-function $H_\xi:=e^{i\xi\cdot x}He^{-i\xi\cdot x}, \xi \in \mathbb R^n, x\in \Lam$, has an analytic continuation from $\mathbb R^n$ to $S_a^n$ if the operators $H_\xi, \xi \in \mathbb R^n$,  have the common domain $\mathcal D(H)$ 
 and, $\forall \psi\in \mathcal D(H),$ the vector function $\xi\to H_{\xi}\psi $,  from $\mathbb R^n$ to $\cH_\A$, is  analytic in $S_a^n$.

\vspace*{.1cm}

8. 
 Condition (AH) implies  $H$ is has the bounded group velocity $i [H, x]$. 
Indeed,  $\n_{\Im \z}\Im H_\z\big|_{\z=0} = i  [H, x]$ and therefore, since $\Im H_\z$ is a bounded harmonic operator - function on $S_a^n$, we have $\|i  [H, x]\|\le C\sup_{\z\in S_{b}^n} \|\Im H_\z\|, b< a$. 

In the opposite direction, it is shown in \cite{SigWu2} (Lemma B.2) that, for small $\mu$,   
\begin{equation}\label{c0}
c(\mu)=\sup\limits_{\hat \eta\in S^{n-1}} \sup \big(\hat\eta\cdot (i[H,x])\big) +O(\mu).
\end{equation}
This allows to compute lower bounds on $c(\mu)$, see Subsection \ref{sec:H}.

\vspace*{.1cm}

9. The results above could be extended to differentiable $H_\xi$, rather than analytic, with the proviso that the exponential tails (spillover) are replaced by those with power decay (long-range models). For this it suffices to use results of \cite{SigWu}, Appendix A, instead of Theorem \ref{thm:MVB-al} below.

\vspace*{.1cm}

10. By Lemma \ref{lem:adj-norm} of Appendix \ref{sec:tech-lem},  \eqref{Gam0-cond} implies that $\G_0 H_{0}$ extends to a trace-class operator:
\begin{align}\label{Gam0-cond'}\|\G_0 H_{0} \|_{S_1^{\AB}}<\infty. 
\end{align}

\vspace*{.1cm}

11. 
Eqs \eqref{Gam-mu} and \eqref{sep-set}  
 show that for any separable state $\G$, every measure $\mu$ associated with it is supported in $\cP_{\rm sep}$, exactly as a classical probability measure would be. As a result,  statistical characteristics 
 of $\G\in \Sigma_{\rm sep}$ can be reproduced by classical means and realized by  a hidden variables model (\cite{Wern}). Correlations of  separable states are said to be classical.

\vspace*{.1cm}

\subsection{Classes of Hamiltonians $H$ and estimates of $c(\mu)$}\label{sec:H} 
Condition~(1) is satisfied for quantum  Hamiltonians of the form
 \begin{align} \label{H}H=T + V,\end{align} acting on $L^2(\Lam)$ or $\ell^2(\Lam)$, depending on $\Lam$, where  $T$ is a self-adjoint operator described below and $V$ is the multiplication operator by a real function $V(x)$ s.t. $H$ is self-adjoint.
 
For $\Lambda=\mathbb Z^n$ (the tight binding approximation), 
  $T$ is a symmetric operator, with exponentially decaying matrix elements, $  |t_{{x, y}}|\leq Ce^{-a|{ x - y}|},$ for some $a>0,$ and $V(x)$ is a bounded function.

For $\Lambda=\mathbb R^n$,  
one can take $T=\om(p)$,  where $\om(k)$ is a real, smooth, positive function on $\mathbb R^n,$ $p:=-i\n$ is the momentum operator 
  and  the potential $V(x)$ is real and s.t.  $H$ is self-adjoint on the domain of $\om(p)$, e.g.  
$V$ is  $\om(p)$-bounded with the relative bound $<1$, i.e.
\begin{align} \label{V-cond}
& \exists \,0\le a <1,\ b>0: \quad
\|Vu\|\le a  \| \om(p) u\|+ b\|u\|,
\end{align}
where $\|\cdot\|$  denotes the norm in $\mathcal H^A=L^2(\Lam)$. 
Operator~\eqref{H} on $L^2(\R^d)$, with $T=\om(p)$, satisfies (AH) if the function $\omega(k)$ has an analytic continuation, $\omega(\zeta),$ from $\mathbb R^n$ to $S_a^n$ and $\im\, \omega(i\eta)$ is a bounded function $\forall i\eta\in S_a$. 

\vspace*{.1cm}

An important example of $\omega(k)$ is the relativistic dispersion law $\omega(k)=\sqrt{|k|^2+m^2}$ with $m>0$, or more generally, $\omega(k)=\sum\limits_{j=1}^N \sqrt{|k_j|^2+m_j^2},$ with $k=(k_1,\cdots,k_N), k_j\in \mathbb R^d, m_j>0$. Thus conditions (H) and (AH) are satisfied for the semi-relativistic $N$-particle quantum Hamiltonian (cf.~\cite{SigWu})
\begin{equation}\label{HN}
H=\sum\limits_{j=1}^N \sqrt{|p_j|^2+m_j^2}+V(x),
\end{equation}
where $x=(x_1,\cdots, x_N), x_j\in \mathbb R^d$, and $p_j=-i\nabla_{x_j}$, $j=1,\cdots,N$, and $V(x_1,\cdots,x_N)$ is a standard $N$-body potential. Hence Theorems \ref{cor:ent-prop1} - \ref{thm:SN-est} hold for discrete and semi-relativistic $N$-body systems.\par

\vspace*{.1cm}

Using \eqref{c0}, we can easily compute $c(\mu)$ in the leading order in $\mu$, e.g., for \eqref{H} on $L^2(\R^d)$, with $T=\om(p)$,  we have 
\begin{equation}\label{c0'}c(\mu)=\sup_{\hat \eta\in S^{n-1}} \sup_{k\in \R^{n}} \big(\hat\eta\cdot \n\om(k) \big)+O(\mu).\end{equation} 
For example, for $\om(k)=\tilde\om(|k|)$, we have $c(\mu)= \sup_{q\ge 0} \tilde\om'(q)+O(\mu)$, which for \eqref{HN} yields $c(\mu)= 1+O(\mu)$, with $1$ being the speed of light in our units.  This bound could be saturated for $V=0$ and initial states localized at very high energies.

\vspace*{.1cm}

For the discrete model  Hamiltonian $H=T + V$ on $\ell^2(\mathbb{Z}^n)$ described above, 
with  the nearest-neighbour hopping,   i.e., $t_{xy}=\tau\delta_{|x-y|=1}$, in 1D,  
\eqref{c0} yields 
\begin{equation}\label{c0''}c(0)=  
\sup (\pm i[H,x])= \tau \sup_{k\in \R}(\pm i\tau (e^{ik}-e^{-ik})=2 \tau \sup_{k\in \R} \sin(k)=2 \tau.\end{equation} 
Passing to physical units, which amounts to replacing  $\tau$ by $\lam\tau/\hbar$, where $\lam$ is the physical lattice spacing, we find  $c(0)=2 \lam\tau/\hbar$. 

\vspace*{.1cm}

  Following \cite{FLSZ},  
 consider a typical 1D optical lattice experiment, 
		 e.g., \cite{Cheneau1, Cheneau2}, which is modelled by   the discrete Hamiltonian with  the nearest-neighbour hopping with an effective hopping amplitude $\tau/\hbar \approx 500  \mathrm{s}^{-1}$ between neighbouring lattice sites  spaced $\lam\approx 500 \mathrm{nm}$ apart. Recalling that $c(0)=2 \lam\tau/\hbar$, this gives 
 $c(0)\approx 5 \cdot 10^{5}$nm/sec. (In the experiment, the propagation observed up to time $t_{\max} \approx 3 \hbar/\tau= 6 \cdot 10^{-3}$sec.) 		
  This bound could be  saturated for $V=0$ by taking the initial quasi-momentum  localized at $k_*=\pi \hbar/2\lam$.

\vspace*{.1cm}

\section{Proofs of  Theorems \ref{thm:Gt-est},  \ref{cor:ent-prop1}(a) and \ref{thm:SN-est}} 
\label{sec:bipart-evol-pf}

\begin{proof}[Proof of Theorem \ref{thm:Gt-est}] Define  the quantum Liouville (von Neumann) operators 
\begin{align}\label{L}
	L \G :=  -i  [H_{\AB},  \G] \quad \text{ and }\ L_{ 0} \G :=  -i  [H_{0},  \G],\end{align}
acting on the space $S_1^{AB}$. We use that $L$ can be written as 
\begin{align}\label{L-deco}	L = L_{ 0} +\hat I , 
\end{align}
where $L$ and $L_{ 0}$ are defined in \eqref{L} and $\hat I \G :=  -i  [I,  \G]$,   
and apply  Duhamel's principle to  \eqref{vNE-eqAB}, to obtain 
\begin{align}\label{Gamma-bar}
	\Gamma_{t}&=e^{t L_{ 0} }\Gamma_{0}+\int_{0}^{t}e^{(t-s) L_{ 0} } \hat I \Gamma_{s}ds.
\end{align}
Using this formula,  
we estimate 
\begin{align}\label{Gt-est1}
   & \|\tilde\chi_X \big( \G_t - e^{t L_{ 0} }\G_0 \big)\|_{S_1^{AB}}
    \leq  \int_{0}^{t}\|\tilde\chi_X e^{s L_{ 0} } \hat I \Gamma_{t-s}ds\|_{S_1^{AB}} ds.
 \end{align} 

\vspace*{.1cm}

We estimate the norm under the integral on the r.h.s. of \eqref{Gt-est1}. Denote $\chi_Y^\otimes  := \chi_Y\otimes \one_B$. Using the definition $\hat I \G :=  -i  [I,  \G]$ and  the relations $e^{r L_{ 0} }\G=e^{-i H_{ 0} r}\G e^{ i H_{ 0}r}$, we write
\begin{align}\label{Gt-est2} \tilde\chi_X e^{s L_{ 0} } \hat I \Gamma &=\tilde\chi_X e^{s L_{ 0} } \big[\chi_Y^\otimes I \Gamma - \Gamma I \chi_Y^\otimes\big] \notag \\
 &= \chi_X^\otimes e^{-i s H_{ 0} }\big[ \chi_Y^\otimes I \Gamma - \Gamma I \chi_Y^\otimes\big]e^{i s H_{ 0} }\chi_X^\otimes.
  \end{align}
Applying the $\mathcal S_1^{AB}$-norm to this identity and using that $\|A\lam\|_{S_1}, \|\lam A\|_{S_1}\le \|A\|\|\lam\|_{S_1}$, we obtain the estimate
\begin{align}\label{Gt-est3}\|\tilde\chi_X  e^{s L_{ 0} } \hat I \Gamma \|_{S_1^{AB}}
\le &\| \chi_X^\otimes e^{-i s H_{ 0} } \chi_Y^\otimes I \Gamma \|_{S_1^{AB}} +\| \Gamma I \chi_Y^\otimes e^{i s H_{ 0} }\chi_X^\otimes\|_{S_1^{AB}}.
  \end{align}
Since \begin{align}\label{Gt-est4}\| \chi_X^\otimes e^{-i s H_{ 0} } \chi_Y^\otimes)\|
&=\|\chi_Y^\otimes e^{i s H_{ 0} }\chi_X^\otimes\| =\| \chi_X e^{-i s H_{\A} } \chi_Y\|,\end{align}
Eq. \eqref{Gt-est3} and the inequalities $\|A\lam\|_{S_1}, \|\lam A\|_{S_1}\le \|A\|\|\lam\|_{S_1}$ give 
\begin{align}\label{Gt-est5}\|\tilde\chi_X & e^{s L_{ 0} } \hat I \Gamma \|_{S_1^{AB}}\le \| \chi_X e^{-i s H_{\A} } \chi_Y\| \big( \| I \Gamma \|_{S_1^{AB}}+\|\Gamma I \|_{S_1^{AB}}\big).
  \end{align}  
Remembering \eqref{I-bndd'}, we take, in what follows, 
\[J:= (H_{0}+1)^{-1}.\]
With this definition and \eqref{I-bndd'}, we have 
\begin{align}\label{I-bndd''}	&\|I J\|=\|J I\|< \infty.\end{align} 
Now, relations \eqref{Gt-est5} and \eqref{I-bndd''} imply  
\begin{align}\label{Gt-est6}\|\tilde\chi_X & e^{s L_{ 0} } \hat I \Gamma \|_{S_1^{AB}}\le 2\| \chi_X e^{-i s H_{\A} } \chi_Y\| \|I J\| \|J^{-1}\Gamma\|_{S_1^{AB}}.
  \end{align}
     This relation, Eq. \eqref{Gt-est1}, condition \eqref{Gam0-cond} ($\|J^{-1}\G_0 \|<\infty $) on $\Gamma_{0}$ and Lemma \ref{lem:G-t-unif-bnd} of Appendix \ref{sec:tech-lem} (giving $\|J^{-1}\G_t \|_{S_1^{AB}}\le C  \|J^{-1}\G_0 \|_{S_1^{AB}}$) yield 
\begin{align}\label{Gt-est7}\|\tilde\chi_X \big( \G_t -& e^{t L_{ 0} }\G_0 \big)\|_{S_1^{AB}}\notag\\
&\le 
C\int_{0}^{t}\| \chi_Xe^{-i H_{ A} s} \chi_Y  \|ds\| I J\|\|J^{-1}\Gamma_{0}\|_{S_1^{AB}}.
 \end{align}
 To estimate the integral on the r. h.s. we use  the following result

\vspace*{.1cm}

\begin{theorem}[\cite{SigWu}]\label{thm:MVB-al} Assume $H$ is self-adjoint and satisfies Condition~
~(AH). Then, for any $\mu\in (0,a)$ and 
for any two disjoint sets $X$ and $Y$ in $\Lambda$, the MQOD $\al_t$ satisfies
\begin{equation}\label{MVB-al}
\| \chi_X e^{-i H_{ A} s} \chi_Y  \|    \leq Ce^{- \mu(d_{XY}-ct)},
\end{equation}
for any $c>c(\mu)$ and some constant $C=C_{n, c, \mu}>0$ depending on $n, c, \mu$. Here  $c(\mu)$, is given  by~\eqref{cmu}.
\end{theorem}

(For an extension of the approach of \cite{SigWu} to many-body setting, see \cite{FLSZ}.)  
 Applying inequality \eqref{MVB-al} to the term $\| \chi_Xe^{-i H_{ A} s} \chi_Y  \|$ in \eqref{Gt-est7}, 
 we arrive at \eqref{Gt-est}.
\end{proof}

\vspace*{.1cm}

All remainders in the rest of this section are estimated in the trace norm of $S_1^\AB$.

\begin{proof}[Proof of Theorem \ref{cor:ent-prop1}(a)]  
Assume $\G_{0}$ is satisfies \eqref{Gam0-cond} and separable outside a set $Q\subset \Lam$, i.e.  $\chi_{Q^c}^\otimes\G_{0}\chi_{Q^c}^\otimes$ is separable. Then, we can write  $\G_{0}= \G_{0}'+ \G_{0}''$, where $\G_{0}'=\chi_{Q^c}^\otimes\G_{0}\chi_{Q^c}^\otimes$ is  
separable and $\G_{0}''$ is of the form
\begin{align}\label{Gam''}\G_{0}'' = {\chi_Q^\otimes }\G_{0}  {\chi_Q^\otimes }&+ {\chi_Q^\otimes }\G_{0} {\chi_{Q^c}^\otimes  } + {\chi_{Q}^\otimes }\G_{0} {\chi_{Q^c}^\otimes }. \end{align}
(We say $\G_{0}''$ is semi-localized in $Q$.)  
 Let $\G_{ 0}'= \sum_i p_i\g_{\A}^i\otimes \g_{\B}^i$.  
 Then $e^{t L_{ 0}}\G_0'= \sum_i p_i\al_t^\A(\g_{\A}^i)\otimes \al_t^\B(\g_{\B}^i)$ and therefore  $\tilde\chi_X e^{t L_{ 0}}\G_0'$ is separable:
\begin{align}\label{Gam't}\tilde\chi_X e^{t L_{ 0}}\G_0'= \sum_i p_i \chi_X \al_t^\A(\g_{\A}^i)  \chi_X \otimes \al_t^\B(\g_{\B}^i). \end{align}

For $\G_{0}''$, by Eqs \eqref{Gam''}, the definitions $\tilde\chi_X(\G)={\chi_X^\otimes}\G{\chi_X^\otimes}$ (see \eqref{tilde-chi}) and $e^{t L_{ 0}}\G= e^{-i H_{0} t}  \G e^{i H_{0} t}$,  we have 
\begin{align}\label{Gam''t2}\tilde\chi_X e^{t L_{ 0}}\G_0'' 
&=\chi_X^\otimes e^{-i H_{0} t}  \G_0'' e^{i H_{0} t} \chi_X^\otimes\notag\\
&  = \chi_{X}^\otimes e^{-i H_{0} t} {\chi_{Q}^\otimes }\G_{0} {\chi_{Q}^\otimes } e^{i H_{0} t} \chi_{X}^\otimes \notag\\
& +  \chi_{X}^\otimes e^{-i H_{0} t} {\chi_{Q}^\otimes }\G_{0} {\chi_{Q^c}^\otimes } e^{i H_{0} t}\chi_{X}^\otimes \notag\\ 
 & +   \chi_{X}^\otimes  e^{-i H_{0} t}
{\chi_{Q^c}^\otimes }\G_{0} {\chi_{Q}^\otimes } e^{i H_{0} t}\chi_{X}^\otimes . 
\end{align}
Now, by the inequalities $\|A\lam\|_{S_1}, \|\lam A\|_{S_1}\le \|A\|\|\lam\|_{S_1}$, and $\|\chi_{Q}^\otimes  e^{-i H_{0} t}\chi_{X}^\otimes \|, $\\
$ \|\chi_{Q^c}^\otimes  e^{-i H_{0} t}\chi_{X}^\otimes \|\le 1$, and the relation $\| \chi_X^\otimes  e^{-i H_{0} s} \chi_Q^\otimes   \| =\| \chi_Q^\otimes  e^{i H_{ 0} s} \chi_X^\otimes   \|$, this  yields
\begin{align}\label{Gam''t-est1}\|\tilde\chi_X e^{t L_{ 0}}\G_0''\|_{\mathcal S_1^\AB} 
 \le &3\|\chi_{X}^\otimes e^{-i H_{0} t}\chi_{Q}^\otimes \|   \|\G_{0} \|_{\mathcal S_1^\AB}. \end{align}
The latter inequality, relation  \eqref{Gt-est4}   and estimates \eqref{MVB-al} and $\|\G_{0} \|_{\mathcal S_1^\AB}= 1$ imply 
\begin{align}\label{Gam''t-est}\|\tilde\chi_X e^{t L_{ 0}}\G_0''\|_{\mathcal S_1^\AB} 
 &    \leq Ce^{- \mu(d_{XQ}-ct)},
\end{align}
for any $c>c(\mu)$ and some constant $C=C_{n, c, \mu}>0$ depending on $n, c, \mu$. 

Now, relations   $\G_{0}= \G_{0}'+ \G_{0}''$ and \eqref{Gam''t-est} imply 
\begin{align}\label{Gam0t-est}\tilde\chi_X e^{t L_{ 0}}\G_0 = \tilde\chi_X e^{t L_{ 0}}\G_0' + O(e^{- \mu(d_{XQ}-ct)}),\end{align}
 which together with \eqref{Gt-est} and the notation $d=\min (d_{XY}, d_{XQ})$, yields
\begin{align}\label{Gamt-est}\tilde\chi_X   \G_t = \tilde\chi_X e^{t L_{ 0}}\G_0'   + O(e^{-2\mu(d -ct)}). 
 \end{align}
According to the definition after Eq. \eqref{un-ent-do}, 
$\G_t$ is  is  separable mod $\kappa$  in $X$, with $\kappa= O(e^{-2\mu(d-ct)})$. This proves  Theorem \ref{cor:ent-prop1}(a).
 \end{proof} 
 
\vspace*{.1cm}

\begin{proof}[Proof of Theorem \ref{thm:SN-est}]
Write $\G_0={\chi_{Q}^\otimes }\G_{0} {\chi_{Q}^\otimes }+{\chi_{Q^c}^\otimes }\G_{0} {\chi_{Q}^\otimes }  +  {\chi_{Q}^\otimes }\G_{0} {\chi_{Q^c}^\otimes }   +   {\chi_{Q^c}^\otimes }\G_{0} {\chi_{Q^c}^\otimes }$ and use the Schwartz and triangle inequalities to estimate 
 \begin{align}\label{G0-ineq} &\G_0 \ge \frac12 {\chi_{Q}^\otimes }\G_{0} {\chi_{Q}^\otimes } - \frac12{\chi_{Q^c}^\otimes }\G_{0} {\chi_{Q^c}^\otimes } =: \G_0'- \G_0''.\end{align}
 By the conditions on $\G_0$, we have that $\G_0', \G_0'' \ge0$ and that $\G_0'$ has SN $k$. Moreover, $\G_0'$ and $\G_0''$ localized in disjoint sets and therefore, in particular, $\G_0'  \G_0''=\G_0'' \G_0'=0$.

\vspace*{.1cm}

Next, proceeding as in \eqref{G0-ineq}, we obtain  
\begin{align}\label{Gam't-est1}  \tilde\chi_X (e^{t L_{ 0} }\G_0')\ge\frac12 e^{t L_{ 0} }\G_0' - \frac12 \tilde\chi_{X^c} e^{t L_{ 0} }\G_{0}' . 
\end{align} 
Now, use $\tilde\chi_{X^c} e^{t L_{ 0} }\G_{0}'=\frac12 \tilde\chi_{X^c} e^{t L_{ 0} }{\tilde\chi_{Q}}\G_{0}=\frac12 V\G_{0}V^* $, where $V:=\chi_{X^c}^\otimes e^{- it H_{ 0} }{\chi_{Q}^\otimes }$,  and then use \eqref{MVB-al} from  Theorem \ref{thm:MVB-al} to estimate $V=O(e^{- \mu(d_{X^cQ}-ct)})$ (cf. \eqref{Gam''t2}-\eqref{Gam0t-est}, arriving at 
\begin{align}\label{Gam't-est2}  \tilde\chi_X (e^{t L_{ 0} }\G_0')\ge\frac12 e^{t L_{ 0} }\G_0' 
+O(e^{- 2\mu(d_{X^cQ}-ct)}). 
\end{align} 

Since $\tilde\chi_X$ and $e^{t L_{ 0}}$ are positivity preserving maps, we have, by \eqref{G0-ineq} and  \eqref{Gam't-est2}. 
\begin{align}\label{nS-est2} 
\tilde\chi_X (e^{t L_{ 0}}\G_0 ) &\ge  \tilde\chi_X (e^{t L_{ 0}}(\G_0' - \G_0'') ) \notag\\
&=  \tilde\chi_X (e^{t L_{ 0}} \G_0' )  -   \tilde\chi_X (e^{t L_{ 0}} \G_0'' ) \notag\\
&\ge  e^{t L_{ 0}} \G_0'   -   e^{t L_{ 0}} \G_0'' - O(e^{- 2\mu(d_{X^cQ}-ct)}). 
\end{align}
By Thm \ref{thm:Gt-est}, $\tilde\chi_X  \G_t = \tilde\chi_X (e^{t L_{ 0}}\G_0 ) + O(e^{- \mu(d_{XY}-ct)})$. The last two relations yield 
\begin{align}\label{nS-est2'} 
\tilde\chi_X (\G_t ) &\ge  \  e^{t L_{ 0}} \G_0'   -   e^{t L_{ 0}} \G_0'' - O(e^{- 2\mu(d'-ct)}),  
\end{align}
where $d'=\min (d_{XY}, d_{X^cQ})\gg1$ and $t<d'/c(\mu)$ 

\vspace*{.1cm}

 Next, note that since $r (\G_0)<\infty$, we have $r(e^{t L_{ 0}}\G_0')<\infty$. Indeed, by elementary properties of the rank, we have $r(e^{t L_{ 0}}\G_0')= r(e^{t L_{ 0}} \frac12 {\chi_{Q}^\otimes }\G_{0} {\chi_{Q}^\otimes })= r( {\chi_{Q}^\otimes }\G_{0} {\chi_{Q}^\otimes })\le   r (\G_0)<\infty$. In this case, one could use

\vspace*{.1cm}

Relation \eqref{nS-est2'}, the facts $e^{t L_{ 0}} \G_0', e^{t L_{ 0}} \G_0''\ge 0$ and $\Tr(e^{t L_{ 0}} (\G_0') e^{t L_{ 0}} (\G_0''))=0$, 
Corollary  \ref{cor:SN-prop} and Eq   
 \eqref{SN-est0} of Appendix \ref{sec:SN-prop} imply that, 
\begin{align}\label{nS-est''}  
 n_S \big(\tilde\chi_X (\G_t)\big)\ge n_S \big(e^{t L_{ 0} }\G_0'\big)=n_S \big(\G_0'\big)=k,  \end{align} 
provided $d'\gg 1$ and $t<d'/c(\mu)$, which is precisely \eqref{SN-est} of  Theorem \ref{thm:SN-est}.  \end{proof}

\vspace*{.2cm}

\section{Conclusion}\label{sec:concl}

In this paper, we establish the entanglement \emph{light-cone structure} in bipartite quantum systems with \emph{localized interactions}. Our approach  offers a rigorous, general framework for analysis of the entanglement propagation.
This structure is complementary to the Lieb--Robinson causal structure  
 for correlations of observables, which does not apply to entanglement. 
 
 To quantify the spatial distribution of entanglement, we introduced the notion of \emph{local entanglement}, which would be especially useful in analysis of the distributed entanglement used in quantum communication. Using this notion, we investigated how this distribution changes with time.   
 
\vspace*{.1cm}

Our bounds are applicable to continuum (including semi-relativistic) as well to short-range lattice quantum systems. 

\vspace*{.2cm}

\bigskip

\paragraph{\bf Acknowledgment}

The  author is grateful to J\"urg Fr\"ohlich, Gian Michele Graf and Marius Lemm for discussions, to  J\'er\'emy Faupin, Marius Lemm, 
Xiaoxu Wu and Jingxuan Zhang, for enjoyable collaboration, and to Chi Kam Wong, for discussions and help with the mnspt.   
His research is supported in part by NSERC Grant NA7901.

\medskip

\paragraph{\bf Declarations}
\begin{itemize}
	\item Conflict of interest: The author has no conflicts of interest to declare that are relevant to the content of this article.
	\item Data availability: Data sharing is not applicable to this article as no datasets were generated or analyzed during the current study.
\end{itemize}

\appendix

 \vspace*{.1cm}

\section{Some technical lemmas}\label{sec:tech-lem}

In this appendix, we present some standard technical results used in the main text.  

\begin{lemma}\label{lem:relat-bnd} Under condition \eqref{I-bndd},  $H_{AB}$ is bounded below.  
\end{lemma}

\begin{proof} Using that $H_{AB}=	H_0+I$, by \eqref{H-AB}, and \eqref{I-bndd}, we obtain
\begin{align}\label{H-AB}	H_{AB}=(H_{0}+1)^{1/2} [\one+ (H_{0}+1)^{-1/2}I &(H_{0}+1)^{-1/2}](H_{0}+1)^{1/2}\notag\\
&\ge (1-\al)(H_{0}+1), 
\end{align}
where $\al:=\|(H_{0}+1)^{-1/2}I (H_{0}+1)^{-1/2}\|< 1$. 
\end{proof}

\begin{lemma}\label{lem:G-t-unif-bnd} Let  \eqref{I-bndd'} holds, let $\G_t$ solve \eqref{vNE-eqAB}  and $\|J^{-1}\G_0 \|<\infty $. 
 Then  $\|J^{-1}\G_t \|_{S_1^{AB}}\le C  \|J^{-1}\G_0 \|_{S_1^{AB}}$, uniformly in $t$.  
\end{lemma}

\begin{proof} Since   $\G_t$ solves \eqref{vNE-eqAB}, we can write $\G_t=e^{L t}\G_0$, where $L$ is defined in  \eqref{L}. Hence, we would like to commute $J$ through $e^{L t}$. Since  the weight $J_{\AB}:= (H_{\AB}+1)^{-1}$ commutes with $ e^{-i H_{ \AB} s} $, we have 
\[J^{-1}\G_t=J^{-1} J_{\AB} e^{-i H_{ \AB} s} J_{\AB}^{-1}\G_0 e^{i H_{ \AB} s}= J^{-1} J_{\AB}  e^{-i H_{ \AB} s} J_{\AB}^{-1} J J^{-1}\G_0 e^{i H_{ \AB} s}.\]
Hence,  it suffices to  show that $J^{-1} J_{\AB}$ and $J_{\AB}^{-1}J $ are bounded. For $J^{-1} J_{\AB}$, we have 
\[J^{-1} J_{\AB}=(H_{0}+1)(H_{\AB}+1)^{-1}=\one - I (H_{\AB}+1)^{-1} =\one - I J J^{-1} J_{\AB},\] 
which, by condition \eqref{I-bndd'}, implies $\|J^{-1} J_{\AB}\|\le 1+ \al\|J^{-1} J_{\AB}\|$, where $\al:=\|I J\|<1$, which shows that $J^{-1} J_{\AB}$is bounded. The boundedness of $J_{\AB}^{-1}J $ follows from the relation 
\[J_{\AB}^{-1}J=(H_{\AB}+1)(H_{0}+1)^{-1}=[\one +I (H_{0}+1)^{-1}]\]
 and the boundedness of $I J$, due to condition \eqref{I-bndd'}. By the above, we have  
 \[\|J^{-1}\G_t\|_{S_1^{AB}}\le \|J^{-1} J_{\AB}\|  \|J_{\AB}^{-1} J\|  \|J^{-1}\G_0 \|_{S_1^{AB}}= C  \|J^{-1}\G_0 \|_{S_1^{AB}},\]
giving the desired statement. 
\end{proof}

\vspace*{.1cm}

For the next result. see \cite{Conv}.

\begin{lemma}\label{lem:adj-norm} $\lam\in S_1\ra \lam^*\in S_1$ and $\|\lam \|_{S_1}=\|\lam^* \|_{S_1}$.  
\end{lemma}

\begin{proof}
Both statements follow from the relations (see \cite{Conv} and \cite{SigWu2}, Eq. (1.24))
   \[\lam\in S_1\Longleftrightarrow \sup_{\|A \|=1}|\Tr(A\lam)|<\infty,\]  $\|\lam \|_{S_1}=\sup_{\|A \|=1}|\Tr(A\lam)|$ and $|\Tr(A\lam)|=|\Tr(A^*\lam^*)|$. \end{proof}  
 
\vspace*{.1cm}

\section{Properties of Schmidt number} 
\label{sec:SN-prop}

In this appendix, we prove a partial extension of the result of \cite{Shir} on the lower semi continuity of $n_S$. 
\begin{prop}\label{prop:SN-prop} For any $\G_1 \ge  0$ of finite rank, $r (\G_1)< \infty$, there is $\eps=\eps(\G_1)$, such that for any $\G, \G_2$ satisfying 
\begin{align} \label{SN-cond}   
 \G\ge \G_1+\G_2,\  \G \ge  0,\    
P\G_2P\ge -\eps P, 
 \end{align}
 where $P$ is the orthogonal projection onto $\Ran \G_1$, we have 
 \begin{align} \label{SN-est4}       
 n_S (\G)\ge n_S (\G_1). \end{align} 
\end{prop}

In the proof below, we use the following elementary lemma 

\begin{lemma}\label{lem:SN-est} For any positive trace-class operators $\G, \G_1, \G_2$ (on a Hilbert space $\cH$), we have   
\begin{align}  
\label{SN-est1}& n_S (\G) \ge n_S (Q\G Q^*),\ \forall \text{ bounded operator  $Q$ acting on } \cH_\A,\\  
\label{SN-est2}&n_S (\G_1)\ge n_S (\G_2),\ \text{  if }\ \G_1 \ge  \G_2.  
 \end{align}
  \end{lemma}

 \vspace*{.2cm}

 \begin{proof}[Proof of Lemma \ref{lem:SN-est}]
 For Eq. \eqref{SN-est1}, we use the relations $Q\G Q^* = \int_{\mathcal P} P_{Q\Psi}\, d\mu(\Psi)= \int_{\mathcal P} P_{\Psi}\, d\mu_Q(\Psi)$, where $\mu_Q$ is defined by $\mu_Q(A)=\mu(Q^{-1}A)$,  
 and $r_{S}(Q\Psi) \le r_{S}(\Psi)$, which together with  
  definition \eqref{nS}  
  of $n_S (\G)$, imply Eq. \eqref{SN-est1}.

 \vspace*{.1cm}

 For Eq. \eqref{SN-est2},  we let $\G_i= \G(\mu_i)$ (not unique), $i=1, 2$, and  use that
 \[\G_1\ge \G_2 \leftrightarrow \mu_1\ge \mu_2 \leftrightarrow n_S (\G_1)\ge n_S (\G_2).\] 
 \end{proof}

 \begin{proof}[Proof of Proposition \ref{prop:SN-prop}]  
 Recall that $P$ is the orthogonal projection onto $\Ran \G$. By  
 \eqref{SN-cond},  
  we have 
 \begin{align}
\label{Gam-est1} &P\G P \ge P(\G_1+\G_2)P = \G_1+P\G_2 P.
 \end{align} 
 Since $r(\G_1)<\infty$, we have $CP\ge \G_1 \ge c P$, for some $C>c> 0$.  
Furthermore, by \eqref{SN-cond},  $P \G_2 P\ge -\eps P$, for some sufficiently small  $\eps=\eps(\G_1)>0$. Hence  
 \[ \G_1 +P\G_2 P \ge (c-\eps)P,\] 
which, together with \eqref{Gam-est1}, yields $P\G P \ge (c-\eps)P$.   Taking $\eps<c$, the latter relation, together with relations $P\ge C^{-1} \G $, gives $P\G P \ge (c-\eps) C^{-1} \G$ and therefore, by Eqs \eqref{SN-est1} and  \eqref{SN-est2},  
  \[n_S (\G) \ge  n_S (P\G P)  
   \ge  n_S (\G),\] 
and so \eqref{SN-est4} follows.  
 \end{proof}

 \vspace*{.1cm}

We have the following elementary lemma: 

\begin{lemma}\label{lem:ortho} Let $\lam, \mu\in S_1$. The relations $\lam, \mu\ge 0$ and $\Tr(\lam  \mu)=0$ imply that  $P_\lam \mu= \mu P_\lam =0$, where $P_\lam$ is the orthogonal projection onto $\Ran \lam$.  \end{lemma}


 \begin{proof}[Proof of Lemma \ref{lem:ortho}]The relations $\lam, \mu\ge 0$ and $\Tr(\lam  \mu)=0$ and the spectral decomposition of $\lam$ and $\mu$ imply that the eigenfunctions of  $\lam$ and $\mu$ are mutually orthogonal, which in turn gives $P_\lam \mu= \mu P_\lam =0$.\end{proof}

\vspace*{.1cm}

 Proposition \ref{prop:SN-prop} and  Lemma \ref{lem:ortho} yield the following 
\begin{corollary}\label{cor:SN-prop'} For any $\G_1 \ge  0$, there is $\eps=\eps(\G_1)$ such that $ \forall \G, \G_2, \G_3,$  satisfying $\G, \G_2\ge  0, \Tr(\G_2 \G_1)=0,\|\G_3\|\le \eps, \G \ge \G_1-\G_2- \G_3$, we have
\begin{align} \label{SN-est4'}       
 n_S (\G)\ge n_S (\G_1).
 \end{align}
  \end{corollary}

 \vspace*{.1cm}

 Finally, we present the following elementary properties of $n_S$:
 
\begin{lemma}\label{lem:SN-est-more} We have the following relations for the Schmidt number $n_S$:

 \begin{align}\label{SN-est0} & n_S (e^{t L_{ 0} }\G)=n_S (\G),\\ 
\label{SN-est3}&n_S (\G_1)+n_S (\G_2)\ge n_S (\G_1+\G_2). \end{align} 
\end{lemma}

 \begin{proof}[Proof of Lemma \ref{lem:SN-est-more}]
 Indeed, by $e^{t L_{ 0} }=e^{t L_{\A} } \otimes e^{t L_{\B} }$, we have $r_{S}(e^{t L_{ 0} }\Psi) = r_{S}(\Psi)$, which, by   definition \eqref{nS}, implies  Eq. \eqref{SN-est0}.  
 
  \vspace*{.1cm} 
 
 We check  Eq. \eqref{SN-est3}. By definition \eqref{nS}, $n_S(\G) := \inf_{\mu\in M_\G} n_S(\mu)$, where 
\[n_S(\mu):= \operatorname*{sup}_{\Psi\in \supp\mu} r_{S}(\Psi).\] 
Next,   let $\G=\frac12(\G_1 +\G_2)$, $\G_i= \G(\mu_i)$ (not unique), $i=1, 2$, and $\mu=\frac12(\mu_1 +\mu_2)$. We have
\begin{align}\label{SN-est3''} n_S(\mu) &= \operatorname*{sup}_{\Psi\in \supp\mu} r_{S}(\Psi)\le  \operatorname*{sup}_{\Psi\in \supp\mu_1} r_{S}(\Psi)+\operatorname*{sup}_{\Psi\in \supp\mu_2} r_{S}(\Psi) \\
&=n_S(\mu_1)+ n_S(\mu_2).\end{align} 
 Since this is true for any $\mu_i\in M_{\G_i}, i=1, 2,$ and since $n_S(\G) := \inf_{\mu\in M_\G} n_S(\mu)$, we conclude that Eq. \eqref{SN-est3} holds.\end{proof}

 \vspace*{.1cm}

 \begin{corollary}\label{cor:SN-prop} For any $\G_1 \ge  0$, there is $\eps=\eps(\G_1)$ such that $ \forall\ \G\ge  0, \G_2, \G_3,$  satisfying 
\begin{align} \label{SN-cond'}\chi\G_1= \G_1\chi = \G_1, \chi\G_2 \chi=0,\ \|\chi\G_3\chi\|_{S_1^\AB}\le \eps, \G \ge \G_1-\G_2- \G_3, \end{align} 
for some bounded operator $\chi$ acting only on $\cH_\A$, we have
\begin{align} \label{SN-est4'}  
 n_S (\G)\ge n_S (\G_1).
 \end{align}
  \end{corollary}   
 \begin{proof}Eq. \eqref{SN-cond'} implies $\chi P=P\chi = P$ and $P\G_2 P=P\chi\G_2 \chi P=0$, which give Eq. \eqref{SN-cond} and therefore \eqref{SN-est4}.\end{proof}


\end{document}